\documentclass{article}
\usepackage{amsmath,amsfonts,amssymb,amsthm}
\usepackage[alldates=year,backend=biber,maxnames=6,doi=true,hyperref=true]{biblatex}
\usepackage{url}
\usepackage{hyperref}
\newtheorem{theorem}{Theorem}
\newtheorem{corollary}{Corollary}

\title{Proof that the Kalman gain minimizes the generalized variance}
\author{Eviatar Bach\footnote{Email: eviatarbach@protonmail.com}}

\renewbibmacro*{doi+eprint+url}{%
    \printfield{doi}%
    \newunit\newblock%
    \iftoggle{bbx:eprint}{%
        \usebibmacro{eprint}%
    }{}%
    \newunit\newblock%
    \iffieldundef{doi}{%
        \usebibmacro{url+urldate}}%
        {}%
    }

\DeclareSourcemap{
  \maps[datatype=bibtex]{
    \map{
      \step[fieldset=issn, null]
    }
  }
}

\DeclareSourcemap{
  \maps[datatype=bibtex]{
    \map{
      \step[fieldset=urldate, null]
    }
  }
}

\DeclareSourcemap{
  \maps[datatype=bibtex]{
    \map{
      \step[fieldset=pagetotal, null]
    }
  }
}

\begin{document}
\maketitle

\begin{abstract}
The optimal gain matrix of the Kalman filter is often derived by minimizing the trace of the posterior covariance matrix. Here, I show that the Kalman gain also minimizes the determinant of the covariance matrix, a quantity known as the generalized variance. When the error distributions are Gaussian, the differential entropy is also minimized.
\end{abstract}

For the classic Kalman filter, the Kalman gain is
\begin{equation}
\mathbf{K} = \mathbf{P}^\text{f}\mathbf{H}^\top(\mathbf{H}\mathbf{P}^\text{f}\mathbf{H}^\top + \mathbf{R})^{-1}\label{eq:gain},
\end{equation}
where $\mathbf{P}^\text{f}$ is the background error covariance matrix (also called the prior error covariance matrix), $\mathbf{R}$ is the observation error covariance matrix, and $\mathbf{H}$ is the linearized observation operator \parencite{asch_data_2016}.

Eq. (\ref{eq:gain}) can be shown to minimize $\operatorname{tr}(\mathbf{P}^\text{a})$, where $\mathbf{P}^\text{a}$ is the analysis error covariance matrix (also called the posterior error covariance matrix). In fact, Eq. (\ref{eq:gain}) is often derived by minimizing $\operatorname{tr}(\mathbf{P}^\text{a})$ \parencite{jazwinski_stochastic_1970,todling_estimation_1999,kalnay_atmospheric_2002,asch_data_2016}. $\mathbf{P}^\text{a}$ is the updated error covariance after observations are assimilated, and can be expressed as
\begin{equation}
\mathbf{P}^\text{a} = \mathbf{(I - KH)}\mathbf{P}^\text{f}\mathbf{(I - KH)^\top + KRK^\top}.
\end{equation}
Note that this formula for the updated covariance matrix applies to any gain $\mathbf{K}$, not necessarily optimal \parencite{asch_data_2016}. It is sometimes known as the Joseph form of the covariance update equation \parencite{todling_estimation_1999}. Here I use the notation common in the atmospheric sciences \parencite{ide_unified_1997}; in other fields, there are a variety of notations, but $\mathbf{P}^\text{f}$ is sometimes written as $\mathbf{P}_{k|k-1}$ and $\mathbf{P}^\text{a}$ as $\mathbf{P}_{k|k}$.

The trace of a covariance matrix is sometimes known as the total variance. However, this measure of the dispersion of the probability distribution does not take into account the cross-covariances between variables. Another scalar measure of dispersion is the \textit{generalized variance}, introduced by Samuel S. Wilks, which is defined as the determinant of the covariance matrix. The generalized variance has been used in a wide variety of fields as a measure of multidimensional scatter \parencite{kocherlakota_generalized_1983,gupta_generalized_2006,chen_insights_2016}.

Since the determinant of a positive definite matrix is lesser or equal to the product of its diagonal elements (this follows from Hadamard's inequality \parencite{horn_matrix_2013}), nonzero off-diagonal elements decrease the generalized variance. This captures the intuition that an error distribution with correlated errors is ``less dispersed'' than the same distribution with uncorrelated errors, a property that the total variance does not possess. For a multivariate Gaussian, contours of the probability density function are ellipsoids of constant Mahalanobis distance from the mean. The volume of these ellipsoids is proportional to the square root of the generalized variance. Increasing correlations between the variables makes the ellipsoid more elongated, and decreases the volume \parencite{wilks_statistical_2019}. The trace and determinant of $\mathbf{P}^\text{a}$ were compared as measures of Kalman filter performance in \parencite{yang_performance_2012}.

One might ask, then, what Kalman gain is obtained by minimizing the generalized variance of the analysis (posterior) distribution. Here, I prove that the Kalman gain in Eq. (\ref{eq:gain}) in fact also minimizes the generalized variance.

\begin{theorem}
The Kalman gain in Eq. (\ref{eq:gain}) minimizes the generalized variance.
\end{theorem}
\begin{proof}
We wish to find the optimal gain $\mathbf{K^*}$ by minimizing the generalized variance of the analysis distribution, defined as the determinant of the analysis covariance matrix:
\begin{align}
\mathbf{K^*} &= {\arg\min}_\mathbf{K} \det(\mathbf{P}^\text{a})\\
&= {\arg\min}_\mathbf{K} \log(\det(\mathbf{P}^\text{a})),
\end{align}
where the second line follows from the fact that $\mathbf{P}^\text{a}$ is positive definite and from the monotonicity of log.

Using an identity of matrix calculus \parencite{petersen_matrix_2012},
\begin{equation}
\text{d}\log(\det(\mathbf{P}^\text{a})) = \operatorname{tr}((\mathbf{P}^\text{a})^{-1} \text{d}\mathbf{P}^\text{a}).\label{eq:jacobi}
\end{equation}

Then $\text{d}\mathbf{P}^\text{a}$ can be expanded:
\begin{equation}
\begin{aligned}
\text{d}\mathbf{P}^\text{a} =& -\mathbf{P}^\text{f}\mathbf{H}^\top \text{d}\mathbf{K}^\top - (\text{d}\mathbf{K}) \mathbf{H}\mathbf{P}^\text{f} + (\text{d}\mathbf{K}) \mathbf{H} \mathbf{P}^\text{f} \mathbf{H^\top K^\top} + \mathbf{KH}\mathbf{P}^\text{f}\mathbf{H}^\top \text{d}\mathbf{K}^\top\\
&\qquad+ (\text{d}\mathbf{K}) \mathbf{R K^\top} + \mathbf{KR}\;\text{d}\mathbf{K}^\top.
\end{aligned}\label{eq:dPa}
\end{equation}

Substituting Eq. (\ref{eq:dPa}) into Eq. (\ref{eq:jacobi}) and simplifying using properties of the trace (distributive, $\operatorname{tr}(\mathbf{A}^\top) = \operatorname{tr}(\mathbf{A})$, and invariant under cyclic permutation when sizes are commensurate) as well as symmetry of the covariance matrices, we obtain
\begin{equation}
\text{d} \log(\det(\mathbf{P}^\text{a})) = \operatorname{tr}\left((2\mathbf{KH\mathbf{P}^\text{f}H^\top} + 2\mathbf{KR} - 2\mathbf{\mathbf{P}^\text{f}H^\top})^\top(\mathbf{P^\text{a}})^{-1}\text{d}\mathbf{K}\right)
\end{equation}

This is equivalent to
\begin{equation}
\frac{\text{d}}{\text{d}\mathbf{K}} \log(\det(\mathbf{P}^\text{a})) = (\mathbf{P}^\text{a})^{-1}(2\mathbf{ KH\mathbf{P}^\text{f}H^\top} + 2 \mathbf{KR} - 2 \mathbf{\mathbf{P}^\text{f}H^\top}).\label{eq:derivative}
\end{equation}

This result can also be verified with \url{http://www.matrixcalculus.org/} \parencite{laue_matrixcalculusorg_2020}.

Setting Eq. (\ref{eq:derivative}) to $\mathbf{0}$ in order to optimize, and since $(\mathbf{P}^\text{a})^{-1} \neq \mathbf{0}$,
\begin{equation}
\mathbf{K^*H\mathbf{P}^\text{f}H^\top + K^*R - \mathbf{P}^\text{f}H^\top = 0}.
\end{equation}

Solving for $\mathbf{K^*}$, we recover Eq. (\ref{eq:gain}).
\end{proof}

\begin{corollary}
Under the assumption of Gaussian error distributions, Eq. (\ref{eq:gain}) also minimizes the differential entropy of the analysis distribution.
\end{corollary}
\begin{proof}
For a multivariate Gaussian with covariance matrix $\pmb{\Sigma}$ and dimension $N$, the differential entropy is \parencite{cover_elements_2006}
\begin{equation}
\frac{N}{2}\log(2\pi e) + \frac{1}{2}\log(\det(\pmb{\Sigma})).
\end{equation}
The first term is a constant and can thus be neglected in the minimization.

If the background and observation error distributions are Gaussian, the analysis error distribution will also be Gaussian. Thus, under the Gaussian assumption, minimizing the generalized variance is equivalent to minimizing the differential entropy.
\end{proof}
Furthermore, for Gaussian error distributions the related quantities of R\'enyi entropy and the entropy power are also minimized \parencite{chen_insights_2016}.

\printbibliography

\end{document}